\documentclass[%
 reprint,
 superscriptaddress,
  amsmath,
  amssymb,
  aps,
  pra,
]{revtex4-2}

\usepackage{algorithm}
\usepackage{algpseudocode}

\usepackage{xcolor}
\definecolor{VibrantBlue}{HTML}{0077BB}
\definecolor{VibrantCyan}{HTML}{33BBEE}
\definecolor{VibrantTeal}{HTML}{009988}
\definecolor{VibrantOrange}{HTML}{EE7733}
\definecolor{VibrantRed}{HTML}{CC3311}
\definecolor{VibrantMagenta}{HTML}{EE3377}
\definecolor{VibrantGrey}{HTML}{BBBBBB}
\usepackage[utf8]{inputenc}
\usepackage{natbib}
\usepackage{graphicx}
\usepackage[colorlinks=true,urlcolor=VibrantBlue, citecolor=VibrantBlue, linkcolor=VibrantMagenta,anchorcolor=VibrantMagenta]{hyperref}
\usepackage{braket}
\usepackage{soul}
\usepackage{bm}

\usepackage{amsmath, amsthm, amssymb}

\usepackage{float}

\theoremstyle{definition}
\newtheorem{mathDef}{Definition}
\newtheorem{mathRem}[mathDef]{Remark}

\theoremstyle{plain}
\newtheorem{mathLem}[mathDef]{Lemma}
\newtheorem{mathProp}[mathDef]{Proposition}
\newtheorem{mathThm}[mathDef]{Theorem}

\newcommand{\N}{\mathbb{N}}
\newcommand{\Z}{\mathbb{Z}}
\newcommand{\R}{\mathbb{R}}
\newcommand{\Prob}{\mathbb{P}}

\newcommand{\Normal}{\mathcal{N}}

\DeclareMathOperator{\IG}{IG}
\DeclareMathOperator{\Sh}{Sh}

\DeclareMathOperator{\Perm}{Perm}

\usepackage{quantikz}

\begin{document}

\title{eXplainable AI for Quantum Machine Learning}
\begin{abstract}
    Parametrized Quantum Circuits (PQCs) enable a novel method for machine learning (ML). However, from a computational point of view they present a challenge to existing eXplainable AI (xAI) methods. On the one hand, measurements on quantum circuits introduce probabilistic errors which impact the convergence of these methods. On the other hand, the phase space of a quantum circuit expands exponentially with the number of qubits, complicating efforts to execute xAI methods in polynomial time. In this paper we will discuss the performance of established xAI methods, such as Baseline SHAP and Integrated Gradients. Using the internal mechanics of PQCs we study ways to speed up their computation.
\end{abstract}
\author{Patrick Steinmüller}
\email{patrick.steinmueller@d-fine.de}
\affiliation{d-fine GmbH}
\author{Tobias Schulz}
\affiliation{d-fine GmbH}
\email{tobias.schulz@d-fine.de}
\author{Ferdinand Graf}
\email{ferdinand.graf@d-fine.de}
\affiliation{d-fine GmbH}
\author{Daniel Herr}
\email{daniel.herr@d-fine.ch}
\affiliation{d-fine AG}
\date{October 2022}

\maketitle

\section{Introduction}
Since most ML models are too complex for humans to properly interpret, methods have been devised to study these models and to provide the ability to understand how models arrive at certain predictions. This explanation is valuable for model developers, who need to understand e.g. the limitations from a model on a global perspective to adjust the model architecture or the training data. It is also valuable for model users, who are interested in the factors that drive specific model predictions. Hence, those methods can improve the overall model quality as well as model acceptance. In addition, governmental initiatives like the European Union's `Artificial Intelligence Act'~\cite{AIAct} require that models applied in high-risk areas (e.g. access to education and essential private services) provide a minimum level of transparency for users, which highlights the importance of xAI methods.  

In this paper, we will study models based on so-called Parameterized Quantum Circuits (PQCs)~\cite{PQCReview}, which are also known as variational quantum circuits or quantum neural networks~\cite{farhi2018classification}, and how xAI methods can be adjusted to these type of models.

\begin{mathDef}[Parametrized Quantum Circuit (PQC)]
\label{def:PQC}
A PQC with \(N\) Qubits and \(n+1\) features consists of \(n+1\) single-qubit rotation gates \(R_{\cdot}^{(j)}(\varphi_i)\), as well as two-qubit entanglement gates. Here, \(j=1,\ldots,N\) denotes on which qubit the gate acts on, \(\cdot=x,y,z\) is a placeholder for the specific axis of rotation, and \(i=1,\ldots,n+1\) indexes the parameters. For our purposes measurements in this circuit are always performed at the end and are in the computational basis.
\end{mathDef}

\begin{mathRem}
In some circuits features are re-uploaded. This means that the feature is encoded by some rotation several times at different parts of the circuit. In the following, we will continue the analysis without re-uploading. However, the re-uploading case is easily realised by just equalizing some features \(x_i = x_j = ...\).
\end{mathRem}

\begin{mathThm}[Output and Learning Behavior]
\label{thm:Output_And_Learning_Behavior}
For a quantum circuit as defined in Definition \ref{def:PQC} the following holds~\cite{Schuld2021}: The output is given as a truncated Fourier series with \(\Omega = \{-1,0,1\}^{n+1}\) as frequency spectrum
\begin{equation}
    f(\varphi) = \sum\limits_{\omega \in\Omega} a_\omega \sin\langle\omega, \varphi\rangle + b_\omega \cos\langle \omega, \varphi\rangle.
\end{equation}
\end{mathThm}

This means that the result of PQCs can be viewed as truncated Fourier Series. In our case the parameters \(a_\omega, b_\omega\) are real numbers and a \(\Z+i\Z\) multiple of \((1/\sqrt{2})^l\) for some exponent of \(l\). If the set of parameters is divided up into features \(\varphi\) and controls \(\theta\), and the feature parameters are allowed to be present repeatedly~\cite{reuploading}, then \(a_\omega, b_\omega\) are trigonometric polynomials in \(\theta\)~\cite{Schuld2021}. The structure of the circuit dictates the amount of control available in training the circuit.

\section{Literature Review}

With the general availability of early quantum computers in the cloud~\cite{IBMQExperience,AWSBracket,Azure} a new field developed in finding useful problems to tackle with the current generation of noisy quantum computers~\cite{Devitt2016,Zeng2017}. At the core of this effort are the aforementioned PQCs, as they have some inherent resistance against systematic errors of the devices. QAOA~\cite{farhi2014quantum} and VQE~\cite{Peruzzo2014} are the most famous algorithms that employ this parameterized approach. Later on, it has been repurposed for various quantum machine learning models~\cite{quantumAdvantage, Herr_2021}. There are still a lot of open questions regarding the effectiveness of training~\cite{McClean2018, Sweke2020stochasticgradient} and their expressivity~\cite{Abbas2021} as machine learning models.

Nevertheless, there is some interesting research into the behaviour of PQCs as machine learning models~\cite{Schuld2021}. We try to build upon this work by introducing explainable AI (xAI) to the field of quantum machine learning. Standard xAI methods might help elucidate the behavior of current PQC-based machine learning models. 

xAI for quantum machine-learning (QML) also seems to have a large overhead with simulation of quantum circuits. Recent advances were demonstrated by beating Google's quantum advantage experiment~\cite{Arute2019} using tensor network-based simulations~\cite{pan2021solving}. There are other approaches such as stabilizer simulations~\cite{PhysRevA.70.052328,Bravyi2019simulationofquantum}, which can simulate large numbers of qubits but scale unfavourably in the number of non-Clifford gates. We use the aforementioned relationship between PQCs and Fourier series in this paper~\cite{Schuld2021}, but expect that other xAI methods can be devised from other simulation approaches.

\section{Review of Model-Agnostic Explainability Methods}

There are many model-agnostic explainability methods currently available. Since our main concern is with PQC based models, we are going to assume that our models are at least continuously differentiable. The methods under consideration in our paper are KernelSHAP~\cite{Shap} and Integrated Gradients\cite{IntegratedGradients}.

In the following, \(f\) denotes the model, \(x\) the input value for which an explanation is sought and \(b\) a base value.

\subsection{Integrated Gradients (IG)}
IG~\cite{IntegratedGradients} is a fast method relying on evaluating the gradient \(\nabla f\) at equidistant points. Denote by \(\gamma_{i;N} = i/N x + (1- i/N)b; i = 0,\ldots, N\) an \(N\) mesh. For the purposes of computing IG efficiently we need to approximate the partial derivative. To do this, we move \(\gamma_{N;i}\) along the \(e\) axis by a shift \(\delta_e\). Using the trapezoid rule and a simple approximation for the partial derivative, we get:

\begin{equation}
    \label{eq:IG_values_definition}
    \begin{split}
        \IG(e)  &= \int\limits_0^1 \frac{\partial f}{\partial x_e} (\gamma(s)) \frac{\mathrm{d} \gamma_e}{\mathrm{d}s}(s) \mathrm{d}s\\
                &\approx \frac {x_e-b_e}{2N} \sum\limits_{i=0}^{N-1} \frac{\partial f}{\partial x_e} (\gamma_{N;i+1}) + \frac{\partial f}{\partial x_e} (\gamma_{N;i})\\
                &= \frac{x_e-b_e}{N}\sum\limits_{i=1}^{N-1} \frac{\partial f}{\partial x_e} (\gamma_{N;i}) \\
                &\quad + \frac{x_e-b_e}{2N}\left(\frac{\partial f}{\partial x_e} (x_e) + \frac{\partial f}{\partial x_e} (b_e)\right)\\
                &\approx \frac{x_e-b_e}{2N\delta_e}\sum\limits_{i=1}^{N-1} f(\gamma_{N;i} + \delta_e) - f(\gamma_{N;i} - \delta_e)
    \end{split}
\end{equation}

In the above derivation the end terms were neglected since for large \(N\) they will vanish to zero.

IG is a very fast method for differentiable models. It follows a path in a straight line from start \(b\) to end \(x\). This makes IG a good choice for large, high-dimensional models, especially in computer vision.

In equation \ref{eq:IG_values_definition} we present two ways of estimating IG values. One with gradients and in the last line with a specific gradient approximation. In settings where differentiable models have good gradient approximations readily available, the first version can be used. If not, the latter can be used as well.

\subsection{SHAP}

We are restricting our discussion of SHAP to Baseline SHAP (BS) which is a sub-variant of KernelSHAP. BS uses a single base value \(b \in \R^{n+1}\) and an input value \(x\in \R^{n+1}\). To get KernelSHAP from BS we can use BS with several base vectors \(b_1,\ldots \) and average the results over those.

Let \(F = [n+1]\) denote the set of features. Let \(e\in F\) the feature for which an explanation is desired. Let \(S\subseteq F\setminus \{e\}\) be a set of features. Denote by \(\overline S\) the opposite set in \(F\setminus \{e\}\) such that \(F=S\cup \overline S \cup \{e\}\). Next we define the intermediate vectors
\begin{equation}
    \label{eq:definition_of_baseline_intermediate_vectors}
    (g_S)_i =
    \begin{cases}
        x_i, & i\in S;\\
        b_i, & \text{otherwise}
    \end{cases}.
\end{equation}

Then the BS value for \(e\) is defined as

\begin{align}
    \label{eq:definition_of_baseline_shap_values}
    \begin{split}
        \Sh_f(e) &= \\
        & \frac 1{n+1}\sum\limits_{S\subseteq F\setminus\{e\}} \frac{|S|!(n-|S|)!}{n!}(f(g_{S\cup \{e\}}) - f(g_S))\\
        &= \frac 1{(n+1)!}\sum\limits_{P\in \Perm(F)} (f(g_{P_e\cup \{e\}}) - f(g_{P_e}))
    \end{split}
\end{align}

In the last line \(\Perm(F)\) denotes the set of permutations of \(F\) and \(P\) a specific permutation of \(F\). \(P_e\) denotes the elements of the permutation that occurred before feature \(e\). I.e. Let \(P=(1,3,2)\) then \(P_2 = \{1,3\}\).

\subsection{Comparison of IG and BS}

Both IG and BS are black-box methods. They do not generally pose many requirements for explainability. The strongest assumption is the requirement of models to be differentiable. This might exclude tree models and neural networks using step-functions as activation function. However, this can be alleviated by using a regularization procedure (i.e. folding against an appropriate derivative of a \(\mathcal{C}^\infty_c\) function).

Both IG and BS are path-dependent feature-perturbation methods. IG samples points along the line spanned by \((b,x)\) while BS samples its points from the set of vertices of the axes-aligned orthotope spanned by \((x,b)\). Here we can see their structural similarity.

BS and IG are model-centric explainability methods. Let \((X,\mathcal{F},\Prob)\) denote the probability space from which data is drawn to train the model \(f\). Both IG and BS will ignore the underlying correlations of the data. BS treats features as independent. In case of linear Models this can be overcome by a correlation correction if the data follows a multivariate Gaussian distribution. IG will assume a correlation of features along the direction of \(x-b\).

BS - SHAP in general - guarantees that features not contributing to the overall model output receive no attribution. The contribution is measured in terms of their marginal impact \(f(g_{S\cup\{e\}})-f(g_S)\). If that difference is always \(0\), e.g. in a linear model which applies a weight \(0\) to that feature, the SHAP value will be zero. IG on the other hand will not guarantee this. Specifically, if \(x_e-b_e\) increase alongside another feature \(x_{e'}-b_{e'}\), and the gradients in both \(e\) and \(e'\) are similar, then both features will have similar contributions.

Both IG and BS are linear in \(f\). This means that computing values for a linear combination of models \(f_i\), it is enough to compute the value for each \(f_i\) and then perform the linear combination.

Both IG and BS feature the concept of a base value. In computer vision tasks this value is often set to \(0\). However, other suitable choices are possible and depend on the problem at hand.

\subsection{Stability of IG and BS}

If a function \(f\) is executed on quantum hardware, the imperfect nature of the device will introduce a significant number of errors. Let \(f\) denote the true \emph{mathematical} function that is approximated by a parametrized quantum circuit. The approximation can be denoted by \(g \approx f + \epsilon\).

For a set of input values \(X = \{x_1, \ldots, x_S\}\), we have a sequence of outputs \(g_i = f(x_i) + \epsilon_i\). We assume that the individual errors \(\epsilon\) are identically distributed with zero-mean \(\mu=0\) and variance \(\sigma^2>0\). In equation \ref{eq:IG_values_definition} we need \(2(N-1)\) function evaluations for computing the IG-value along a single axis. We need \(2(N-1)n\) to compute IG-values for all features.

Let \(S\) be the set of evaluation points: \(S=\{\gamma_{N;i}-\delta_e, \gamma_{N;i}+\delta_e\}\) and \(s=|S|\) its size. Using \ref{eq:IG_values_definition} and plugging in our definition for \(g\):

\begin{align*}
    \IG_e(g) &= \IG_e(f) + \frac{x_e-b_e}{2N\delta_e}\sum\limits_{i=1}^{N-1} \epsilon_{\gamma_{N;i}+\delta_e} - \epsilon_{\gamma_{N;i}-\delta_e}\\
    \IG_e(g) &= \IG_e(f) + \underbrace{\frac{x_e-b_e}{2N\delta_e}\sum_{i\in S} \epsilon_i}_{=:X_{N}}
\end{align*}

Using the central limit theorem (CLT), the error \(X_{N}\) converges in distribution to \(\Normal\left(0, \frac{\sigma^2}{N}\right)\).

We employ the same approach in analysing the stability of BS. In equation \ref{eq:definition_of_baseline_shap_values} two definitions of BS were given. The summation over all permutations has duplicate function evaluations over single points, while the definition using the powerset of \(F\setminus \{e\}\) has the individual evaluations scaled by a term, weighing extremal points stronger than others.

Let us first consider the (inefficient) implementation via the permutation sum. If the function is evaluated at the same point for every permutation anew, then each error for the sum is independent. Using linearity of the SHAP values we have: \(\Sh_g(e) = \Sh_f(e) + \Sh_\epsilon(e)\) and further:

\begin{align*}
    \Sh_\epsilon(e) = \frac 1{(n+1)!}\sum\limits_{P\in \Perm(F)} \epsilon_F - \epsilon'_F.
\end{align*}

Using the CLT we get \(X_n = \Sh_\epsilon(e) \sim \Normal\left(0,\frac{2\sigma^2}{(n+1)!}\right)\). The inefficient algorithm reduces the impact of any error, while the number of function evaluations explodes significantly. When evaluating the error with the standard formulation we face a different scenario:

\begin{align}
    \Sh_\epsilon(e) &= \frac 1{(n+1)!}\sum\limits_{S\subseteq F\setminus \{e\}} (|S|!(n-|S|)!)( \epsilon_S - \epsilon'_S)\\
     &= \frac 1{n+1}\sum\limits_{s=0}^n\sum\limits_{S\subseteq F\setminus \{e\}; |S|=s} \binom{n}{s}^{-1}( \epsilon_S - \epsilon'_S).
\end{align}

The inner sum, sums over \(\binom{n}{s}\) many sets. For \(s=0\) or \(s=n\) the number of summations is \(1\). This exactly fits our intuition that errors at the extremal points of the spanned box have the strongest impact on the outcome. However, even in this case we can expect some convergence against a normal distribution. As the number of features increases we expect \(X_n = \Sh_\epsilon(e) \sim \Normal\left(0,\frac{\sigma^2}{n}\right)\).

\section{Extending SHAP to PQCs}

\subsection{Mathematical Results}

We are going to employ the result from Theorem \ref{thm:Output_And_Learning_Behavior} to find a (faster) extension of BS for PQCs. First we recall some definitions.

\begin{mathDef}[Multivariate Polynomials]
Let \(V\) be a \(K\) vector-space, where \(K\) is a field. We define a polynomial \(p\) as an element of \(\bigoplus_{i\geq 0} S^i(V)\) where \(S^i(V)\) denotes the space of symmetric multivariate forms. Elements of \(S^i(V)\) are also called monomials of order \(i\). Each monomial of order \(i\) can be represented by a totally symmetric tensor \(\Lambda_i\). Totally symmetric tensors allow for a rank-\(1\) decomposition as follows
\begin{equation}
    \label{eq:general_rank_one_decomposition}
    \Lambda_i = \sum_{j=1}^p \lambda_j v_j^{\otimes i}; \alpha_j\in K, v_j\in V.
\end{equation}
\end{mathDef}

This form is easily stored on a computer. To show the relationship between rank one compositions and Fourier Series note the simple Taylor expansion.

\begin{equation}
\label{eq:fourier_rank_one_decomposition}
    \exp(-i\langle \omega,\varphi\rangle) = \sum\limits_{k\geq 0} \frac {(-1)^k i^k \langle \omega,\varphi\rangle^k}{k!} \approx \sum\limits_{k=0}^K \frac {(-1)^k i^k \langle \omega,\varphi\rangle^k}{k!}
\end{equation}

Expanding all terms in a Fourier Series shows that the wave vectors \(\omega_i^{\otimes k}\) represent a natural choice for the rank-\(1\) decomposition. Note that Taylor expansions are not the only choice here. Other polynomial approximation algorithms can be used as well. Chebyshev approximations will reduce the required polynomial order by about half.

\begin{mathProp}[PolynomialSHAP]
    \label{prop:BaselinePolynomialSHAP}
    Let \(\Lambda\) be a totally symmetric tensor of order \(r\) with a rank-one decomposition as in \ref{eq:general_rank_one_decomposition}. Let \(x, b\in \R^m\) be  input and baseline, respectively. Then the BS values are given by:
    \begin{equation}
        \begin{split}
            \Sh(e) = 2\sum\limits_{\substack{j + m + k = r\\0\leq j,m,k\leq r\\ m\text{ odd}; k\text{ even}}} &\binom{r}{j;m;k}\sum\limits_{\substack{\gamma \in \N^n\\|\gamma|=k}} \frac 1 {l(\gamma)+1}\binom{k}{\gamma}\\ &\sum\limits_{i=1}^p \left\langle v_i, M\right\rangle^{j}\left\langle v_i, \frac{\Delta_e}2\right\rangle^m \\
            &\lambda_i \prod\limits_{h=1}^n (v_{ih}\alpha_h)^{\gamma_h}.
        \end{split}
        \label{eq:BaselinePolynomialSHAP}
    \end{equation}
    Here the following definitions are made:
    \begin{enumerate}
        \item \(\N\) denotes the natural numbers with \(0\);
        \item \(M = (x+b)/2\); \(\Delta = x-b\); \(\Delta_e\) is the vector \(\Delta_{i}=0\) for \(i\neq e\) and \(\Delta_e\) otherwise; \(\alpha_h = ((x+b)/2)_h\);
        \item \(l(\gamma)\) denotes the number of odd indices in \(\gamma\).
    \end{enumerate}
\end{mathProp}

The proof of this statement can be found in the supplement. The formula above extends the concept of Linear SHAP to Multivariate polynomial functions and runs in \(\mathcal{O}(m^r)\), where \(m\) denotes the number of features and \(r\) the degree of the polynomial.

\subsection{Algorithmic Description and Runtime}

Algorithm \ref{alg:qSHAP} describes the qSHAP routine.

\begin{algorithm}[H]
    \caption{Compute SHAP Values for PQCs \(f\)}
    \label{alg:qSHAP}
    \begin{algorithmic}[1]
        \Require \(f\) function describing the PQC, \(n\) number of features, \(N_i\) number of times feature \(i\in\{1,\ldots, n\}\) is encoded. \(S\) set of samples in data-space, \(k\) order of polynomial approximation.
        \State Set of admissible wave-vectors \(\Omega=\prod\limits_{i=1}^n [-N_i, N_i]_{\Z}\).
        \For{\(\omega\in \Omega\)}
            \State Compute Fourier Coefficients \(a_\omega, b_\omega\) with sample points \(S\)
            \State Using taylors theorem, compute polynomial extension up to order \(k\).
        \EndFor
        \State Set \(\Sh(e)\gets 0\) for all features \(1\leq e \leq n\).
        \For{\(\omega \in \Omega\)}
            \For{(k=1) to \(K\):}
                \State Use PolynomialSHAP to compute contributions: \(\Sh_{\omega, k}(e)\).
                \State \(\Sh(e)\gets \Sh(e) + \Sh_{\omega, k}(e)\).
            \EndFor
        \EndFor
    \end{algorithmic}
\end{algorithm}

Essentially the algorithm has three steps:
\begin{enumerate}
    \item Evaluate the PQC on random sample points.
    \item Compute the Fourier decomposition. And form a multivariate polynomial approximation in rank-one form.
    \item Use the PolynomialSHAP Algorithm to compute the contributions from each term.
\end{enumerate}

\subsection{Open Issues and Further Direction}

Assuming access to quantum hardware the limiting issue for this approach is to obtain a reasonable rank-\(1\) approximation. Currently, having access to the full set of Fourier Modes is required.

The algorithm does not scale exponentially in the number of qubits used for the PQC, but rather in the number of features present in the algorithm. At present no efficient algorithm for Fast Fourier transformation is known to the authors that can take advantage of the way PQCs activate their Fourier terms. The Fourier terms seem to be concentrated close to the border of \(\Omega\).

Subsequent steps of the algorithm run at most in polynomial time. 

\section{Experiments}

\begin{figure}
    \includegraphics[width=\linewidth]{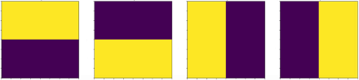}
    \caption{For our examples we use the bars and stripes data set. All possibilities of two by two pixel images are shown in this figure. The two images on the left are examples of stripes and the two images on the right are examples of bars. Our xAI models are applied to classifiers trained to distinguish bars from stripes.}
    \label{fig:BarsStripes}
\end{figure}

Our experiments are conducted with the 2 by 2 \emph{Bars and Stripes} problem. All possible images are shown in Figure~\ref{fig:BarsStripes}. The classifiers were trained to distinguish between the images of bars and images of stripes. The three classifiers each used a different parameterized quantum circuit shown in figures \ref{fig:single_qubit}, \ref{fig:two_qubits} and \ref{fig:four_qubits}. We are going to compare the KernelSHAP implementation of BS, IG, and our approximate qSHAP algorithm both with noise and without noise on these trained PQC classifiers.

For all xAI algorithms and backends we use the same model with the same parameters. That means that the training for each model was performed on classical hardware. The simulations in this case had no noise. Then we compared the results of the different xAI methods across a range of different hardware: classical simulation with no noise, classical simulation with shot-noise, classical simulation with shot-noise and a noise model and execution on a real QPU. To reach optimal model performance in each scenario it is generally recommended to perform the training on the same hardware, on which the model is used later on. However, in this case each model would have different model parameters and the results of the xAI methods would be harder to compare. For this reason we deployed the classically trained model across all hardware options.

Other techniques used for noise reduction on NISQ machines are mentioned in \cite{parallelAcceleration}. We did not train the models for each of the noise models further. Instead we can use the different performance of the explainers as a check of robustness against noise.

\subsection{Single-qubit classifier}
\begin{figure}
    \includegraphics[width=\linewidth]{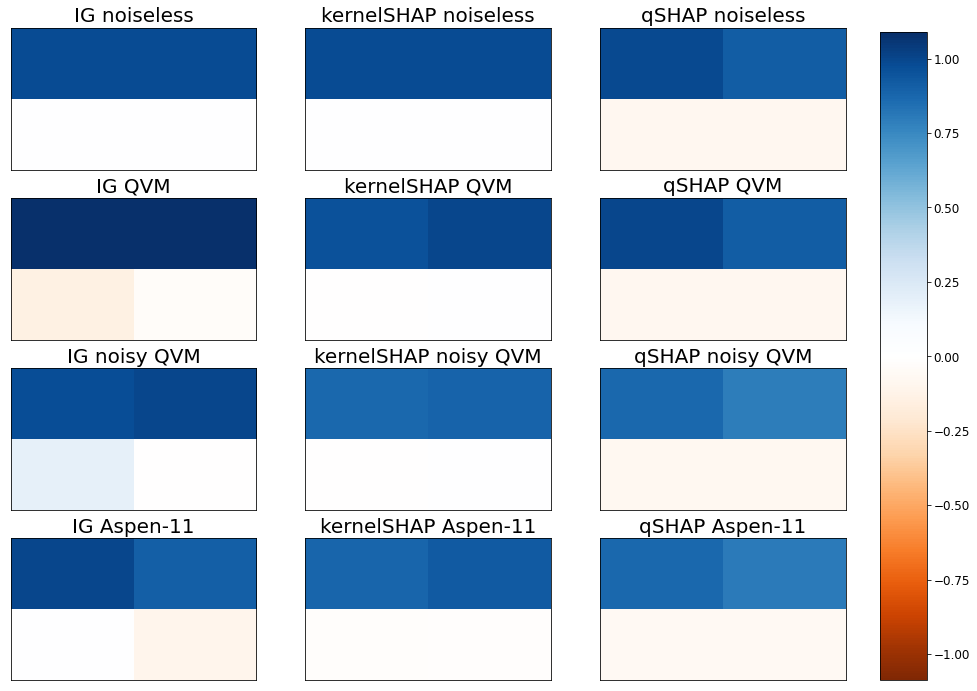}
    \caption{A comparison of the different xAI methods applied to the simple \emph{single-qubit classifier}. The circuit for the classifier is given in Figure \ref{fig:single_qubit}. The higher the value for a pixel the more influence it has on the classification. On the horizontal axis the different explainability methods are plotted (Integrated Gradients, KernelSHAP, and qSHAP). Vertically we run the model on different backends. The topmost row uses a state-vector simulator (no noise except for shot noise). The two following rows use different settings for the Rigetti simulator. One with a generic noise model (second row), and one with the noise model of the device (third row). The bottom row shows the results from a run on Rigetti's Aspen-11 quantum computer.}
    \label{fig:reduced}
\end{figure}

The single qubit classifier uses the same qubit to upload the top two pixels of the 2 by 2 Bars and Stripes images. The circuit for this classifier is given in Figure~\ref{fig:single_qubit} and allows it to effectively learn an XOR operation of the two inputs. These inputs are the two top pixels of a bars and stripes image. If they are the same, the image shows stripes and if they are different, the image shows bars.

Looking at the xAI models in the noiseless case, we see that the two upper pixel have an effect towards the model prediction. The bottom pixel roughly have a predictive value of zero. This is the expected behavior as only the upper pixels are used for the model.

With the inclusion of the first noise model this effect does not seem to change with KernelSHAP and our qSHAP method. IG, on the other hand, already seems to struggle due to the noise of this model. With increasing noise this discrepancy becomes more and more pronounced. The noise model on the real hardware seems to be the most severe, causing a lot of artifacts for IG and to a lesser extent for KernelSHAP. The qSHAP method seems to be the most robust against this noise. The noise even leads to negative contributions for a pixel. This is most pronounced in the bottom left pixel for IG, yet this pixel is not used in the classifier.

\subsection{Two-qubit classifier}
\begin{figure}
    \includegraphics[width=\linewidth]{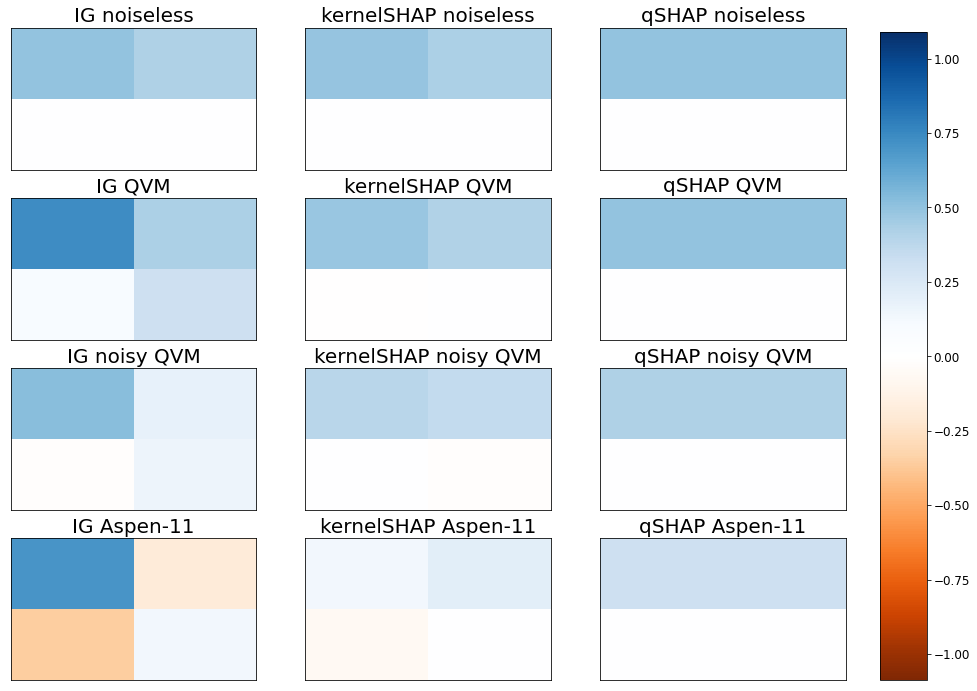}
    \caption{A comparison of the different xAI methods applied to the \emph{two-qubit classifier}. The circuit for the classifier is given in Figure \ref{fig:two_qubits}. The two upper pixels are used in the classifier. For more information on the different rows and columns of this Figure refer to the description in ~\ref{fig:single_qubit}. Any explainability attributions in the lower two pixels are likely due to noise effects (e.g. shot noise).}
    \label{fig:topRow}
\end{figure}
The two-qubit model is slightly more complicated compared to the one-qubit model. It uses a circuit with two qubits, where each qubit is initialized with a rotation dependent on one of the upper two pixels of the bars and stripes image. Again the bottom two pixel are not used in the classifier.

The different xAI models fare similarly for this classifier. This can potentially be explained by the fact that noise does not have as large an influence as for the one-qubit classifier.

\subsection{Four-qubit classifier}
\begin{figure}
    \includegraphics[width=\linewidth]{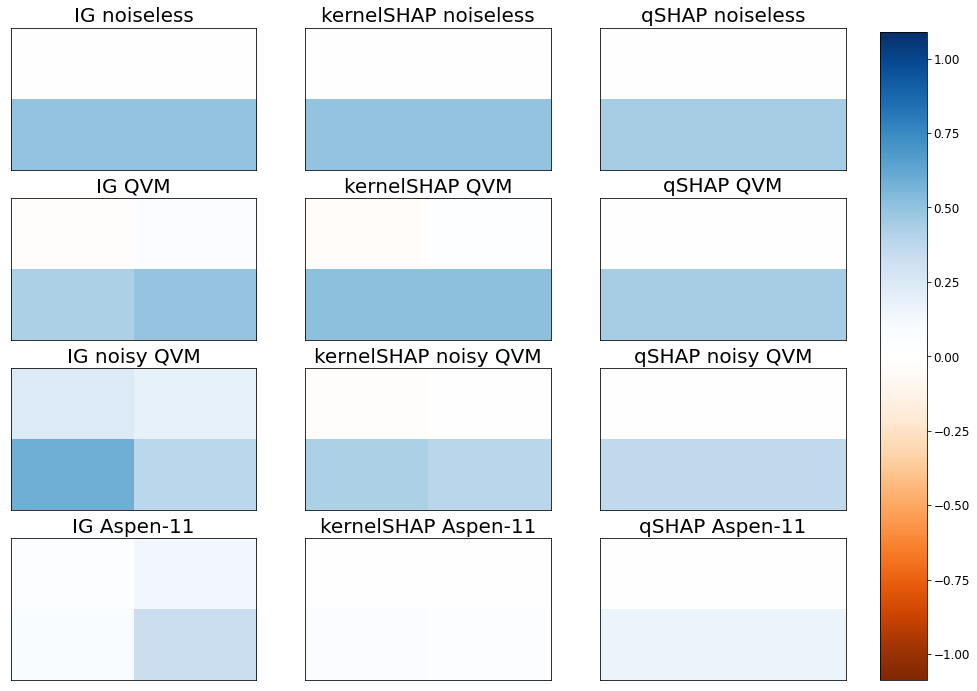}
    \caption{A comparison of the different xAI methods applied to the \emph{four-qubit classifier}. The circuit for the classifier is given in Figure \ref{fig:four_qubits}. All pixel are used as input to the classifier and only the measurement result of the last qubit is used for the classification. For more information on the different rows and columns of this Figure refer to the description in ~\ref{fig:single_qubit}.}
    \label{fig:full}
\end{figure}
Each pixel of a bars and stripes image is encoded in a different qubit in the four-qubit classifier. After two cycles of parameterized rotations and entanglement operations the fourth qubit is then measured to obtain the classification.

This is the first classifier where all qubits could affect the classification. This means that we cannot interpret a nonzero influence of all pixels as instabilities of the xAI models due to e.g. noise. Nevertheless, in the noiseless case, all xAI models show that only the lower two qubits have an effect on the classification.

\subsection{Classification Performance vs. Explainability}

In xAI research one is often faced between achieving a high classification performance or achieving fast and "good" explanations. To a limited degree we see this issue appear here as well. The four-qubit model has the highest number of trainable parameters, meaning that in theory, the space of functions it can encode is much larger. The classification problem presented here is very simple and can be solved in various ways by only taking a subset of the features. In our four-qubit circuit, the qubits that are furthest away from the measurement qubit experience a lesser weighing than the other qubits. This behaviour is similar to classical neural networks where it is harder for the training algorithms to correctly adjust the weights. Furthermore, the quantum nature of PQCs introduces noise, whose impact increases with the size and complexity of the circuit.

\subsection{Runtime for qSHAP}

When discussing the runtime of this algorithm we should discuss three parts here:
\begin{enumerate}
    \item The Sampling Step.
    \item The Fourier Step.
    \item The Polynomial Step.
\end{enumerate}

During the \textit{Sampling Step} the quantum circuit is sampled \(S\) times. The amount of samples required depends on the dimension of the feature space, the circuit noise, etc. This step is linear in the number of samples used. The necessary number of samples should be drawn from convergence criteria from the stochastic integration methods used. However, in practice the size of the circuit is also an issue. There are many steps involved in compiling the logical circuit to the circuit that is evidently used by the hardware. A more complex circuit will introduce a higher overhead, resulting in less evaluations per minute.

Assuming a probabilistic rate of convergence of \(\mathcal{O}((n+1)^{-\frac 12})\) the number of circuit evaluations should scale like \(\mathcal{O}(\epsilon^{-1} (n+1)^2)\).

In our runs we used between 250 and 300 circuit evaluations and the below table lists the total time it took to gather the data:

\begin{table}[H]
\centering
\resizebox{\columnwidth}{!}{
    \begin{tabular}{|c|c|r|}
    \hline
         Algorithm & Samples & Time [s] \\
         \hline
         Single Qubit & 250 & 399.526641 \\
         Two Qubit & 250 & 497.247195 \\
         Four Qubit & 300 & 682.531422 \\
         \hline
    \end{tabular}}
        \caption{Table comparing the runtimes recorded for our experimentes of qSHAP on the Rigetti-M-11 hardware.}
        \label{tab:qSHAP_runtimes}
\end{table}

Starting with the \textit{Fourier Step} we do not use the QPU anylonger. This step is used to estimate which Fourier modes are activated. Importantly, the search space scales exponentially with the number of features and not the number of qubits (as the phase space does). Thus the limiting factor is the number of features in the model. Let \(N_\omega\) denote the number of fourier modes found.

Lastly, the \textit{Polynomial Step} runs in polynomial time in the number of features. Importantly, this step is also depending on \(N_\omega\).

\subsection{Runtime for IG}

IG was executed with 20 nodes on the line from base value to the input value. For each point we had to evaluate the circuit eight times to approximate the partial derivatives. The runtimes are shown below.

\begin{table}[H]
\centering
\resizebox{\columnwidth}{!}{
    \begin{tabular}{|c|c|r|}
    \hline
         Algorithm & Samples & Time [sec] \\
         \hline
         Single Qubit & 20 & 240.778704 \\
         Two Qubit & 20 & 261.569914 \\
         Four Qubit & 20 & 343.405340 \\
         \hline
    \end{tabular}}
        \caption{Table comparing the runtimes recorded for our experimentes of qSHAP on the Rigetti-M-11 hardware.}
        \label{tab:IG_runtimes}
\end{table}

\section{Conclusion}
In this paper we have applied two black-box explainers to nascent quantum machine learning models. The drawback of these standard black-box xAI models is their exponential scaling and their behaviour under the noise of NISQ devices. As a first step in resolving these issues we propose qSHAP, which is an xAI method specifically designed for PQCs.

Please note that this algorithm does not resolve issues related to the exponential scaling. However, the algorithm does not scale exponentially with respect to the number of qubits, but with respect to the number of features.

For qSHAP, the adverse effect of noise is considerably reduced compared to the other explainers.

While there is a large overlap between xAI for PQCs and the general simulation of circuits, we expect some interesting new approaches to arrive from the recent breakthroughs in simulating circuits used to show quantum advantage~\cite{quantumAdvantage}.

Future work can also be focused on finding efficient algorithms for approximating PQCs with multivariate polynomials in rank-one decomposition. Together with the PolynomialSHAP subroutine, this would result in an efficient algorithm for computing SHAP values for PQCs.

\section{Acknowledgements}
This work was supported by the German Federal Ministry for Economic Affairs and Climate Action through project PlanQK (01MK20005D).

Special thanks goes to Till Duesberg who provided welcome vetting of many aspects of the initial drafts and Dr. Cristian Grozea who provided constructive criticism throughout. We also want to thank the d-fine internal reviewers Dr. Daniel Ohl de Mello and Dr. Ulf Menzler.
\bibliography{references}

\onecolumngrid

\appendix

\section{Diagrams of the Circuits}

This section contains the circuits that we used for our experiments.

\begin{figure}[H]
    \begin{quantikz}
    \lstick{\(\ket{0}\)} & \gate{R_x(\varphi_0)} & \gate{R_y(\theta_0)} & \gate{R_x(\varphi_1)} & \qw
    \end{quantikz}
    \caption{Single-Qubit circuit that can learn the 2 by 2 Bars and Stripes by  using two pixels in a row or in a column. This equates to learning XOR.}
    \label{fig:single_qubit}
\end{figure}
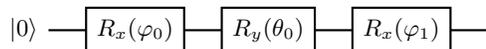

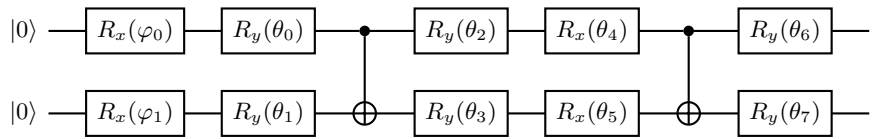
\begin{figure}[H]
    \begin{quantikz}
    \lstick{\(\ket{0}\)} & \gate{R_x(\varphi_0)} & \gate{R_y(\theta_0)} & \ctrl{1} & \gate{R_y(\theta_2)} & \gate{R_x(\theta_4)} & \ctrl{1} & \gate{R_y(\theta_6)} & \qw\\
    \lstick{\(\ket{0}\)} & \gate{R_x(\varphi_1)} & \gate{R_y(\theta_1)} & \targ{} & \gate{R_y(\theta_3)} & \gate{R_x(\theta_5)} & \targ{} & \gate{R_y(\theta_7)} & \qw
    \end{quantikz}
    \caption{Two-Qubit circuit that can learn the 2 by 2 Bars and Stripes by  using two pixels in a row or in a column. This equates to learning XOR.}
    \label{fig:two_qubits}
\end{figure}

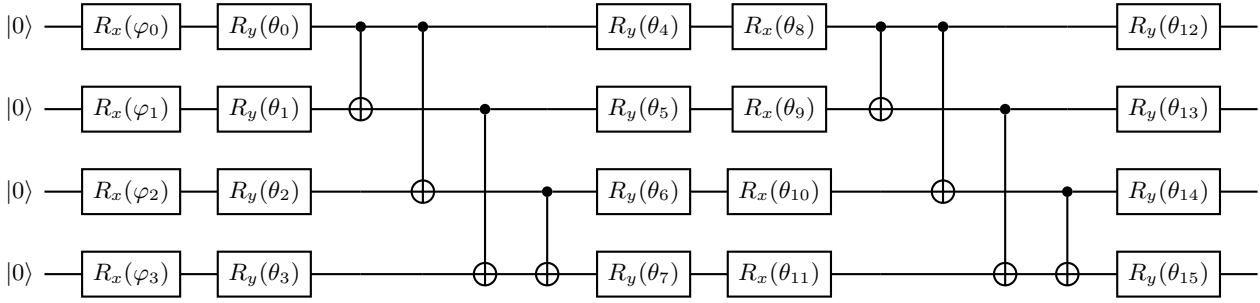
\begin{figure}[H]
    \begin{quantikz}
    \lstick{\(\ket{0}\)} & \gate{R_x(\varphi_0)} & \gate{R_y(\theta_0)} & \ctrl{1} & \ctrl{2} & \qw & \qw & \gate{R_y(\theta_4)} & \gate{R_x(\theta_8)} & \ctrl{1} & \ctrl{2} & \qw & \qw & \gate{R_y(\theta_{12})} & \qw\\
    \lstick{\(\ket{0}\)} & \gate{R_x(\varphi_1)} & \gate{R_y(\theta_1)} & \targ{} & \qw & \ctrl{2} & \qw & \gate{R_y(\theta_5)} & \gate{R_x(\theta_9)} & \targ{} & \qw & \ctrl{2} & \qw & \gate{R_y(\theta_{13})} & \qw\\
    \lstick{\(\ket{0}\)} & \gate{R_x(\varphi_2)} & \gate{R_y(\theta_2)} & \qw & \targ{} & \qw & \ctrl{1} & \gate{R_y(\theta_6)} & \gate{R_x(\theta_{10})} & \qw & \targ{} & \qw & \ctrl{1} & \gate{R_y(\theta_{14})} & \qw\\
    \lstick{\(\ket{0}\)} & \gate{R_x(\varphi_3)} & \gate{R_y(\theta_3)} & \qw & \qw & \targ{} & \targ{} & \gate{R_y(\theta_7)} & \gate{R_x(\theta_{11})} & \qw & \qw & \targ{} & \targ{} & \gate{R_y(\theta_{15})} & \qw
    \end{quantikz}
    \caption{Four-Qubit circuit that can learn the 2 by 2 Bars and Stripes by using all pixels in the image.}
    \label{fig:four_qubits}
\end{figure}

\section{Proof of Proposition \ref{prop:BaselinePolynomialSHAP}}

In this section we detail the proof for Proposition \ref{prop:BaselinePolynomialSHAP}. The following Lemma demonstrates our approach for the quadratic model and motivates our general approach.

\begin{mathLem}
	Consider the BS problem with input \(x\) and baseline \(b\) for a quadratic form model \(f:\R^{n+1}\rightarrow \R, x\mapsto x^TCx\) where \(C\in \R^{(n+1)\times (n+1)}\) denotes a symmetric matrix. With mean \(M=\frac 12(x+b)\) and difference \(\Delta = (x-b)\). The shap values can be computed as:
	\begin{equation}
		\label{eq:BS_SHAP_quadratic_form}
		\Sh(e) = 2 M^T C \Delta_e
	\end{equation}
\end{mathLem}
\begin{proof}
	Using \(b^TCa = a^TCb\) for all \(a, b\in \R^{n+1}\):
	\begin{align*}
		\Sh(e) &= \sum\limits_{S\subseteq [n+1]\setminus\{e\}} \frac{|S|!|\overline S|!}{(n+1)!} [(x_S + x_e + b_{\overline S})^TC (x_S + x_e + b_{\overline S}) - (x_S + b_e + b_{\overline S})^TC (x_S + b_e + b_{\overline S})]\\
		&=\sum\limits_{S\subseteq [n+1]\setminus\{e\}} \frac{|S|!|\overline S|!}{(n+1)!} 2\left(x_S+b_{\overline S} + \frac12(x_e+b_e)\right)^TC(x_e-b_e)
	\end{align*}
	Evaluating the left side of the form:
	\begin{align*}
		&\quad 2\left(x_S+b_{\overline S} + \frac12(x_e+b_e)\right)\\
		&= 2\left(\frac {x_S+x_e+x_{\overline S} + b_S+b_e+b_{\overline S}}2 + \frac{x_S-b_S - (x_{\overline S}-b_{\overline S})}2\right)\\
		&= 2\left(M + \frac 12 (\Delta_S-\Delta_{\overline S})\right)
	\end{align*}
	Thus:
	\begin{equation*}
		\Sh(e) = \sum\limits_{S\subseteq [n+1]\setminus\{e\}} \frac{|S|!|\overline S|!}{(n+1)!} [2M^TC\Delta_e + (\Delta_S-\Delta_{\overline S})^TC\Delta_e]
	\end{equation*}
	Since we are summing over all subsets of \([n+1]\setminus \{e\}\) we have an even number of summands. The transform \(S\mapsto \overline{S}\) is bijective and maps the summand:
	\begin{equation*}
		\frac{|S|!|\overline S|!}{(n+1)!}(\Delta_S-\Delta_{\overline S})^TC\Delta_e \mapsto -\frac{|S|!|\overline S|!}{(n+1)!}(\Delta_S-\Delta_{\overline S})^TC\Delta_e.
	\end{equation*}
	Thus the second sum cancels out and we are left with:
	\begin{equation*}
		\Sh(e)=2M^TC\Delta_e\sum\limits_{S\subseteq [n+1]\setminus\{e\}} \frac{|S|!|\overline S|!}{(n+1)!} 
	\end{equation*}
	We evaluate the last sum. Instead of directly summing over all subsets \(S\) we sum over the subsets of given size \(s\) and then over all sizes. Assume we have subset of \(S\) of size \(s\). There are \(\binom{n}{s}\) ways to select \(s\) elements out of the set \([n+1]\setminus\{e\}\). Thus we can transform the sum:
	\begin{align*}
		\sum\limits_{S\subseteq [n+1]\setminus\{e\}} \frac{|S|!|\overline S|!}{(n+1)!} &= \sum\limits_{s=0}^{n}\sum\limits_{\substack{S\subseteq [n+1]\setminus\{e\}\\|S|=s}} \frac{s!(n-s)!}{(n+1)!}\\
		&= \sum\limits_{s=0}^{n}\frac{s!(n-s)!}{(n+1)!}\binom{n}{s}\\
		&= \sum\limits_{s=0}^{n} \frac 1{n+1} = 1
	\end{align*}
\end{proof}

For two vector spaces \(V, W\) and bases \(a_1,\ldots, a_n\) and \(b_1,\ldots, b_m\) we can represent any \(r\)-multilinear map \(\Lambda\) by the values of \(\Lambda(a_{i_1},\ldots,a_{i_r})\) in the base of \(W\). These values are denoted by \(\Lambda_{i_1,\ldots,i_r;j}\) for \(1\leq j\leq m; 1\leq i_1,\ldots,i_r\leq n\). If \(W\) is one-dimensional. This representation is thus naturally in the tensor product written as:
\begin{align*}
	\Lambda(v_1,\ldots,v_r) &= \sum\limits_{1\leq j_1\ldots j_r \leq n} \alpha_{1;j_1}\cdots \alpha_{r;j_r} \lambda(a_{i_1}\otimes\cdots\otimes a_{j_r})\\
	&= \sum\limits_{j=1}^m \sum\limits_{1\leq j_1\ldots j_r \leq n} \Lambda_{i_1,\ldots,i_r;j} \alpha_{1;j_1}\cdots \alpha_{r;j_r} \beta_j b_j
\end{align*}
If furthermore \(\Lambda\) is symmetric then for any permutation \(\pi\) we have \(\Lambda_{i_1,\ldots,i_r;j} = \Lambda_{i_{\pi(1)},\ldots,i_{\pi(r)};j}\) and that:
\begin{align*}
	\Lambda(v_1,\ldots,v_r) &= \sum\limits_{1\leq j_1\ldots j_r \leq n} \alpha_{1;j_1}\cdots \alpha_{r;j_r} \lambda(a_{i_1}\otimes\cdots\otimes a_{j_r})\\
	&= \sum\limits_{j=1}^m \sum\limits_{1\leq j_1\leq \ldots\leq  j_r \leq n} \sum\limits_{\substack{\beta\in\N^r\\|\beta|=r}}\binom{n}{\beta} \Lambda_{i_1,\ldots,i_r;j} \alpha_{1;j_1}\cdots \alpha_{r;j_r} \lambda(a_{i_1}\vee\cdots\vee a_{j_r})\\
	&= \sum\limits_{j=1}^m \sum\limits_{1\leq j_1\leq \ldots\leq  j_r \leq n} \sum\limits_{\substack{\beta\in\N^r\\|\beta|=r}}\binom{n}{\beta} \Lambda_{i_1,\ldots,i_r;j} \alpha_{1;j_1}\cdots \alpha_{r;j_r} \beta_j b_j
\end{align*}

\(r\)-multilinear maps are our generalisation for monomials. Below we introduce some non-standard notations for convenience:
\begin{enumerate}
	\item Let \(x\in \R^n\), then for any \(l\geq 0\) we let \(x^{\otimes l}\) denote the \(l\)-fold self-tensorproduct of \(x\);
	\item For \(\beta\in \N^n\) and \(|\beta|=\beta_1+\ldots+\beta_n = r\) we define \(\binom{r}{\beta}=\frac{r!}{\beta_1!\cdots \beta_k!}\);
	\item Let \(\beta \in \N^n\) be an n-dimensional index vector of rank \(r=|\beta|\). For \(\alpha\in\R^n\) we have 
	\begin{equation*}
		\alpha^{\vee r} = \sum\limits_{|\beta|=r}\underbrace{\binom{r}{\beta}\alpha_1^{\beta_1}\cdots \alpha_n^{\beta_n}}_{=:(\alpha^{\vee r})_{\beta}} e_1^{\vee \beta_1}\vee \cdots \vee e_n^{\vee \beta_n};
	\end{equation*}
	\item We introduce a generalization of the scalar product we call \(\odot\): \(\Lambda\odot M^{\otimes l} = \sum\limits_{i_{r-l+1}=1,\ldots,i_r=1}^n \Lambda_{\ldots i_{r-l+1},\ldots,i_r} M_{i_{r-l+1},\ldots,i_s}\) collapses the last dimension of the tensor \(\Lambda\). Since \(\Lambda\) is symmetric multilinear the exact sequence of collapse is irrelevant.
\end{enumerate}
With the above notation we can rewrite the multinomial expansion as:
\begin{align*}
	(x_1+\ldots+x_n)^k &= \sum\limits_{|\beta|=k} (x^{\vee k})_\beta = \sum\limits_{|\beta|=k} \binom{k}{\beta} x_1^{\beta_1}\cdots x_n^{\beta_n}\\
						&= \sum\limits_{1\leq \beta_1\ldots\beta_n\leq k} \underbrace{x_{\beta_1}\cdots x_{\beta_n}}_{=:(x^{\otimes k})_\beta} = \sum\limits_{\beta \in [1,n]^k} (x^{\otimes k})_\beta.
\end{align*}
Furthermore for symmetric multilinear forms we have:
\begin{equation*}
	(\Lambda \odot M^{\otimes r}) \odot x^{\otimes s} = \Lambda \odot (M^{\otimes r}\otimes x^{\otimes s}) = \Lambda \odot (x^{\otimes s}\otimes M^{\otimes r})
\end{equation*}
For convenience we drop the paranthesis. 

\begin{mathProp}
	Let \(V\) be a \(K\) vector space and \(\Lambda:V^r\mapsto K\) be a symmetric \(r\)-multilinear form. Let \(a_1, \ldots, a_r\in V\). A rank one decomposition of \(\Lambda\) is given by a natural number \(p\), vectors \(v_1,\ldots,v_p\in V\) and coefficients \(\lambda_1,\ldots,\lambda_p\in V\) such that:
	\begin{equation}
		\Lambda = \sum\limits_{i=1}^p \lambda_i v_i^{\otimes r}.
	\end{equation}
	Such a decomposition is not unique. The minimum number \(p\) for such a decomposition is called the (generalized) rank of \(\Lambda\). The evaluation of \(\Lambda\) is given by scalar products as follows:
	\begin{equation}
		\Lambda(a_1,\ldots,a_r) = \sum\limits_{i=1}^p \lambda_i \prod\limits_{j=1}^r \langle v_i, a_j\rangle
	\end{equation}
\end{mathProp}

\begin{proof}
    With that in mind we can start evaluating BS values for a symmetric \(r\)-form:
    \begin{align*}
    	\Sh(e) &=\sum\limits_{S\subseteq [n+1]\setminus\{e\}} \frac{|S|!|\overline S|!}{(n+1)!} (\Lambda \odot (\underbrace{x_S+b_{\overline S} + x_e}_{=:A})^{\otimes r} - \Lambda \odot (\underbrace{x_S+b_{\overline S} + b_e}_{=:A'})^{\otimes r})
    \end{align*}
    The terms in the parenthesis can be rearranged as follows:
    \begin{align*}
    	A &= x_S+b_{\overline S} + x_e = \frac {x_S + x_e + x_{\overline S} + b_S + b_e + b_{\overline S} + x_S - b_S + x_e - b_e + b_{\overline S} - x_{\overline S}}2 \\
    	&= M + \frac {\Delta_e}2 + \frac{\Delta_S - \Delta_{\overline S}}2\\
    	A' &= M - \frac {\Delta_e}2 + \frac{\Delta_S - \Delta_{\overline S}}2
    \end{align*}
    Thus yielding:
    \begin{align*}
    	& \Sh(e) \\
    	&= \sum\limits_{S\subseteq [n+1]\setminus\{e\}} \frac{|S|!|\overline S|!}{(n+1)!}\sum\limits_{l,k} \binom{r}{l;k} (1-(-1)^l)\Lambda \odot M^{\otimes (r-m-k)}\otimes \left(\frac{\Delta_e}2\right)^{\otimes m} \otimes \left(\frac{\Delta_S - \Delta_{\overline S}}2\right)^{\otimes k}\\
    	&= \sum\limits_{m,k} \binom{r}{m;k} (1-(-1)^l)\Lambda \odot M^{\otimes (r-m-k)}\otimes \left(\frac{\Delta_e}2\right)^{\otimes m} \otimes \left(\sum\limits_{S\subseteq [n+1]\setminus\{e\}} \frac{|S|!|\overline S|!}{(n+1)!} \left(\frac{\Delta_S - \Delta_{\overline S}}2\right)^{\otimes k}\right)
    \end{align*}
    
    The above summands are zero if \(l\) is even or \(k\) is odd.
    
    	The preceeding discussion motivates us to evaluate the the term
	\begin{align*}
		 & \sum\limits_{S\subseteq [n+1]\setminus\{e\}} \frac{|S|!|\overline S|!}{(n+1)!}\left(\frac {\Delta_S-\Delta_{\overline S}}2\right)^{\otimes k}\\
		=& \sum\limits_{S\subseteq [n+1]\setminus\{e\}} \frac{|S|!|\overline S|!}{(n+1)!}\sum\limits_{\beta\in[1,n]^k}\epsilon(S)(\alpha^{\otimes k})_{\beta} e^{\otimes \beta}.
	\end{align*}
	Where \(\alpha_i = \frac 12(\Delta_S-\Delta_{\overline S})_i\) if \(i\in S\) and \(\alpha_i = -\frac 12(\Delta_S-\Delta_{\overline S})_i\) if \(i\in \overline S\). Let \(\epsilon_i\) denote a sign that is \(+1\) if \(i\in S\) and \(-1\) otherwise. Thus
	\begin{align*}
		& \sum\limits_{S\subseteq [n+1]\setminus\{e\}} \frac{|S|!|\overline S|!}{(n+1)!}\sum\limits_{\beta\in[1,n]^k}\epsilon(S)(\alpha^{\otimes k})_{\beta} e^{\otimes \beta}\\
		= &  \sum\limits_{\beta\in[1,n]^k}(\alpha^{\otimes k})_{\beta}e^{\otimes \beta}\sum\limits_{S\subseteq [n+1]\setminus\{e\}} \frac{|S|!|\overline S|!}{(n+1)!}\prod\limits_{\beta_i\in \overline S} (-1) .
	\end{align*}
	
	In the following we denote by \(L, 2l=|L|\) the set indices with odd exponents for a given monomial, whose number must be even. Let \(o = 2l - \overline o= 2l - |\overline S \cap L|\):
	\begin{align*}
		& \sum\limits_{\beta\in[1,n]^k}(\alpha^{\otimes k})_{\beta}e^{\otimes \beta}\sum\limits_{S\subseteq [n+1]\setminus\{e\}} \frac{|S|!|\overline S|!}{(n+1)!}\prod\limits_{\beta_i\in \overline S} (-1) \\
		= & \sum\limits_{\beta\in[1,n]^k} (\alpha^{\otimes k})_{\beta}e^{\otimes \beta}\sum\limits_{o=0}^{2l}\sum\limits_{\substack{S\subseteq [n+1]\setminus\{e\}\\|\overline S\cap L| = 2l -o}} \frac{|S|!|\overline S|!}{(n+1)!}\prod\limits_{\beta_i\in \overline S} (-1)
	\end{align*}
	
	The indices of odd exponents distribute over \(S\) and \(\overline S\), where \(o\) denotes the number of hits in \(S\) and \(\overline o\) in \(\overline S\), respectively. The signature \((o,\overline o), o+\overline o = 2l\) completely determines the sign of the \(\epsilon\) product. If \(o\) is even, then it is positive, otherwise negative. Thus the sum becomes:
	\begin{equation*}
		\prod\limits_{\beta_i\in\overline S} (-1) = (-1)^{2l - o}
	\end{equation*}
	Since only those terms of the monomial matter, that have odd exponents, we get a combinatorial problem as follows:
	\begin{quote}
		Given a signature of \((o, \overline o)\), how many ways are there to arrange \(S\) and \(\overline S\) to hit that signature?
	\end{quote}
	There need to be at least \(o\) elements in \(S\), but also \(\overline o=2l - o\) in \(\overline S\). The rest of the elements can be chosen freely.
	\begin{equation*}
		\binom{2l}{o}\binom{n-2l}{s-o}.
	\end{equation*}

	This can be used to evaluate the inner sum now:
	\begin{align*}
		& \sum\limits_{\substack{S\subseteq [n+1]\setminus\{e\}\\|S\cap L|=o}} \frac{|S|!|\overline S|!}{(n+1)!} (-1)^{2l-o}\\
		=&\sum\limits_{s=o}^{n-(2l-o)} \frac{s!(n-s)!}{(n+1)!} \binom{2l}{o}\binom{n-2l}{s-o} (-1)^{2l-o}\\
		=&\binom{2l}{o}\frac{(n-2l)!}{(n+1)!}(-1)^{2l-o}\sum\limits_{s=o}^{n-(2l-o)} \frac{s!}{(s-o)!}\frac{(n-s)!}{(n-s-(2l-o))!}\\
		=&\binom{2l}{o}\frac{(n-2l)!}{(n+1)!}(-1)^{2l-o}\sum\limits_{s=o}^{n-(2l-o)} s\cdots(s-o+1)\cdot (n-s)\cdots(n-(2l-o) + 1 - s)
	\end{align*}

	We notice the polynomial in the sum would also be zero in the case of \(0\leq s \leq o-1\) and \(n-(2l-o) + 1\leq s \leq n\) and so we can extend the range of the sum:

	\begin{equation*}
		\binom{2l}{o}\frac{(n-2l)!}{(n+1)!}(-1)^{2l-o}\sum\limits_{s=0}^{n} s\cdots(s-o+1)\cdot (n-s)\cdots(n-(2l-o) + 1 - s).
	\end{equation*}

	We claim now that the sum over the polynomial can always be solved and will yield a similar result depending on the parameters \(o, 2l, s, k\).
	
	\begin{mathDef}[Derived Sequences]
    	Let \(v:\Z \rightarrow \R\) be a sequence. Then the following are derived sequences:
    	\begin{itemize}
    		\item The (first) difference sequence: \(\Delta v:\Z \rightarrow \R, s\mapsto (\Delta v)(s):=v(s) - v(s-1)\);
    		\item The \(n\)-th difference sequence is recursively defined: \((\Delta^{n+1} v)(s) := (\Delta(\Delta^{n} v))(s)\);
    	\end{itemize}
    \end{mathDef}

	With
	\begin{align*}
		(\Delta u)(s) &= s\cdots (s-o+1)\\
		v(s) &= (n-s)\cdots (n - (2l - o) + 1 -s)
	\end{align*}
	We have:
	\begin{enumerate}
		\item \((\Delta u)(s)\) has roots at \(0, \ldots, o-1\),
		\item \(v(s)\) has roots at \(n - (2l - o) + 1, \ldots, n\),
		\item \(u\) could have the form \(u(s) = \frac 1{o + 1} (s+1)\cdots (s-o+1)\) and has roots at \(-1,\ldots o-1\),
		\item \(\Delta v(s) = -(2l - o) \cdot (n - 1 - s)\cdots (n - (2l - o) + 1 - s)\) and has roots at \(n - (2l - o) + 1, \ldots, n-1\).
	\end{enumerate}

	With the above convention we can formulate:
	\begin{align*}
		\sum\limits_{s=0}^n (\Delta u)(s)v(s) &= \sum\limits_{s=0}^n u(s)v(s) - \sum\limits_{s=0}^n u(s-1)v(s)\\
		&= \sum\limits_{s=0}^n u(s)v(s) - \sum\limits_{s=0}^n u(s-1)(v(s) - v(s-1)) - \sum\limits_{s=0}^n u(s-1)v(s-1)\\
		&= \sum\limits_{s=0}^n u(s)v(s) - \sum\limits_{s=-1}^{n-1} u(s)v(s) - \sum\limits_{s=0}^n u(s-1)\Delta v(s) \\
		&= u(n)\underbrace{v(n)}_{=0} - \underbrace{u(-1)}_{=0}v(-1) - \sum\limits_{s=-1}^{n-1} u(s)\Delta v(s)\\
		&=- \sum\limits_{s=-1}^{n-1} u(s)\Delta v(s+1).
	\end{align*}
	More generally we can formulate for \(1\leq j \leq 2l-o\):
	\begin{align*}
			\sum\limits_{s=0}^n (\Delta^j u)(s)v(s) &= - \sum\limits_{s=-1}^{n-1} (\Delta^{j-1} u)(s)\Delta v(s)\\
			&= (-1)^j \sum\limits_{s=-j}^{n-j} u(s) (\Delta ^j v)(s+j).
	\end{align*}

	So to evaluate the sum:
	\begin{equation*}
		\sum\limits_{s=0}^n s\cdots (s-o+1)\cdot (n-s)\cdots (n-(2l-o) + 1 -s)
	\end{equation*}
	we use the slightly different settings:
	\begin{align*}
		(\Delta^{2l - o} u)(s) &= s\cdots (s-o+1)\\
		v(s) &= (n-s)\cdots (n - (2l - o) + 1 -s).
	\end{align*}
	The \(i\)-th integral of \(\Delta^{2l - o} u\) has the form:
	\begin{equation*}
		(\Delta^{2l - o + i} u)(s) = \frac 1{(o+1)\cdots (o+i)} (s+i)\cdots (s-o+1)
	\end{equation*}
	having roots at \(-i,\ldots, o-1\).
	The \(i-th\) derivate of \(v\) has the form:
	\begin{equation*}
		(\Delta^i v)(s) = (-1)^i(2l - o)\cdots(2l - o - i + 1) (n - i - s)\cdots (n - (2l - o) + 1 - s)
	\end{equation*}
	Using this we can show that:
	\begin{align*}
		& \sum\limits_{s=0}^n (\Delta^{2l - o} u)(s)v(s + 2l - o) \\
		&= (-1)^{2l - o} \sum\limits_{s = - (2l - o)}^{n - (2l - o)} u(s)(\Delta^{2l - o}v)(s)\\
		&= (-1)^{2l - o}\sum\limits_{s = - (2l - o)}^{n - (2l - o)} \frac {o!}{(2l)!} (s + (2l -o))\cdots(s - o +1) (-1)^{2l - o}(2l - o)!\\
		&= \binom{2l}o^{-1}\sum\limits_{s = - (2l - o)}^{n - (2l - o)}  (s + (2l -o))
		\cdots(s - o + 1)\\
		&=\binom{2l}o^{-1}\left[\frac 1{2l + 1}(s + (2l -o)+1)\cdots(s - o + 1)\right]_{- (2l -o) - 1}^{n - (2l - o)}\\
		&= \frac 1{2l + 1}\binom{2l}o^{-1} (n+1)\cdots(n-2l + 1)
	\end{align*}
	Which means that 
	\begin{align*}
		\quad \quad &\binom{2l}{o}\frac{(n-2l)!}{(n+1)!}(-1)^{2l-o}\sum\limits_{s=0}^{n} s\cdots(s-o+1)\cdot (n-s)\cdots(n-(2l-o) + 1 - s) \\
		=&\binom{2l}{o}\frac{(n-2l)!}{(n+1)!} (-1)^{2l-o}  \frac 1{2l + 1}\binom{2l}o^{-1} (n+1)\cdots(n-2l + 1)\\
		=&\frac{(n-2l)!}{(2l+1)(n+1)!} (n+1)\cdots (n-2l+1) (-1)^{2l - o}\\
		=&\frac {(-1)^{2l -o}}{2l + 1}
	\end{align*}
	To evaluate the rest of the sums:
	\begin{equation}
		\sum\limits_{\beta \in [1,n]^r}(\alpha^{\otimes k})_{\beta}e^{\otimes \beta} \sum\limits_{o=0}^{2l} \frac {(-1)^{2l -o}}{2l + 1} = 
		\sum\limits_{\beta \in [1,n]^r}\frac 1 {2l+1}(\alpha^{\otimes k})_{\beta}e^{\otimes \beta}
	\end{equation}
	For an efficient evaluation of the remaining sum we have to know how \(l\) depends on \(\beta\). 
	
	It is easy to see, that given such a decomposition, storage and evaluation of such multilinear forms is much faster. Combining all of our results leeds to the following expression:
	\begin{align*}
		  & \Sh(e) \\
		= & 2\sum\limits_{\substack{j + m + k = r\\0\leq j,l,k\leq r\\ l\text{ odd}; k\text{ even}}} \binom{r}{j;m;k}\sum\limits_{\substack{\gamma \in \N^n\\|\gamma|=k}} \frac 1 {2l+1} \binom{k}{\gamma}\sum\limits_{i=1}^p \left\langle v_i, M\right\rangle^{j}\left\langle v_i, \frac{\Delta_e}2\right\rangle^m \lambda_i \prod\limits_{h=1}^n (v_{ih}\alpha_h)^{\gamma_h}
	\end{align*}
\end{proof}
\end{document}